\documentclass[a4paper,12pt]{article}
\usepackage{amsfonts}
\usepackage{amsmath}
\usepackage{amsthm}
\usepackage{amssymb}
\usepackage{url}
\usepackage{mathrsfs}
\usepackage{mathtools}
\usepackage{hhline}
\usepackage{ulem}
\usepackage{float}
\usepackage{bm}
\usepackage{textcomp}
\usepackage{enumitem}
\usepackage{physics}
\usepackage[longnamesfirst]{natbib}
\usepackage{xcolor}

\usepackage{fourier}
\usepackage[T1]{fontenc}
\DeclareMathAlphabet{\mathcal}{OMS}{cmsy}{m}{n}

\usepackage[margin=1in]{geometry}

\setlist[itemize]{itemsep=0pt}
\setlist[enumerate]{itemsep=0pt}

\usepackage[mathlines]{lineno}
\newtheorem{lemma}{Lemma}

\newtheorem{theorem}{Theorem}
\newtheorem{proposition}{Proposition}

\mathtoolsset{showonlyrefs=true}

\DeclareMathOperator*{\argmax}{arg\,max}

\title{Condorcet-loser dominance among scoring rules\footnote{We thank Toyotaka Sakai and Koichi Tadenuma for their helpful comments. Doi acknowledges the financial support from JST SPRING Grant Number JPMJSP2123.}}

\author{Ryoga Doi\footnote{Graduate School of Economics, Keio University, Tokyo 108-8345, Japan. Email: doiryouga@keio.jp.}~ and Kensei Nakamura\footnote{Graduate School of Economics, Hitotsubashi University, Tokyo 186-0004, Japan. Email: kensei.nakamura.econ@gmail.com.}}

\begin{document}

\maketitle
\begin{abstract}
  This paper studies a dominance relation among scoring rules with respect to avoiding the selection of the Condorcet loser.
  In a voting model with three or more alternatives, we say that a scoring rule $f$ \textit{Condorcet-loser-dominates (CL-dominates)} another scoring rule $g$ if the set of profiles where $f$ selects a Condorcet loser is a proper subset of the set where $g$ does.
  We show that the Borda rule not only CL-dominates all other scoring rules, but also is the only scoring rule that CL-dominates some scoring rule.
  \vspace{2mm}
\end{abstract}

\section{Introduction}

The plurality rule is a widely used method of collective decision-making, for example, in national elections and in voting within committees and meetings.
This rule selects the alternative that is ranked first by the largest number of voters.
Despite its popularity, the plurality rule sometimes selects a Condorcet loser---an alternative that loses a pairwise majority vote against every other alternative.
This weakness was famously pointed out by \citet{borda1784} through the following example.%
\footnote{For an English translation, see \citet{mclean1994condorcet}.}
Consider $21$ voters and three alternatives, A, B, and C.
Suppose that eight voters rank A above B above C, seven rank B above C above A, and six rank C above B above A (see Table \ref{tab:bordaex}).
In this case, the plurality rule chooses A, although A is beaten by B and C in pairwise majority comparisons.
Borda criticized this consequence, remarking that ``[w]e might compare them to two athletes who, having exhausted themselves competing against one another, are beaten by a third who is actually weakest of all.''

\begin{table}[h]
  \centering
  \caption{Preference profile with 21 voters and 3 alternatives (A, B, and C)}
  \label{tab:bordaex}
  \begin{tabular}{|c|c|c|} \hline
    8 voters & 7 voters & 6 voters \\ \hline
    A & B & C \\ \hline
    B & C & B \\ \hline
    C & A & A \\ \hline
  \end{tabular}
\end{table}

Borda then proposed a rule that assigns $2$ points to the first-ranked alternative, $1$ point to the second, and $0$ points to the third, selecting the alternative with the highest total score as the winner.
This rule is now known as the \textit{Borda rule}.
In contrast to the plurality rule, the Borda rule never chooses a Condorcet loser as a winner.
Indeed, in the above example, B is chosen since A gets $16$ points, B gets $28$ points, and C gets $19$ points.

Both the plurality rule and the Borda rule belong to the class of \textit{scoring rules}.
These rules assign scores to alternatives according to their positions in the rankings and compare them by their total scores.
As shown by \citet{Fishburn1976Borda}, all scoring rules other than the Borda rule  select a Condorcet loser at some preference profile (see also \citet{okamoto2019borda}).
Thus, the Borda rule holds a special status among scoring rules.

For the case of three alternatives, by normalizing the score assigned to the first rank to $1$ and the score for the third rank to $0$, each scoring rule can be identified with the score assigned to the second rank, denoted by $s_2 \in [0, 1]$. \citet{lepelley2000scoring} calculated the probability that each scoring rule selects a Condorcet loser when it exists. Table \ref{tab:lepelley} presents part of their results, where $s_2=0$ corresponds to the plurality rule and $s_2=0.5$ corresponds to the Borda rule. From this table, it can be observed that the closer a rule is to the Borda rule, the less likely it is to select the Condorcet loser.

\begin{table}[h]
  \centering
  \caption{Probability of selecting a Condorcet loser (101 voters)}
  \label{tab:lepelley}
  \begin{tabular}{|c|c|c|c|c|c|c|c|c|c|c|c|} \hline
    $s_2$ & 0 & 0.1 & 0.2 & 0.3 & 0.4 & 0.5 & 0.6 & 0.7 & 0.8 & 0.9 & 1.0 \\ \hline
    Prob. & 0.047 & 0.024 & 0.013 & 0.005 & 0.001 & 0 & 0.001 & 0.005 & 0.014 & 0.025 & 0.045 \\ \hline
  \end{tabular}
\end{table}

From this observation, one might think that there could be a relationship such that if a scoring rule similar to the Borda rule (e.g., the rule with $s_2 = 0.499$) chooses the Condorcet loser at some preference profile, then another scoring rule far from the Borda rule (e.g., the rule with $s_2 = 0.001$) does so;
or equivalently, whenever the latter does not choose a Condorcet loser, then the former does not. 
If such a relationship holds, we can conclude that a rule closer to the Borda rule is ``better'' than a rule farther from it in terms of avoiding the Condorcet loser.
The objective of this paper is to examine whether we can rank scoring rules under this criterion.

Formally, given a set of at least three alternatives, we say that a scoring rule $f$ \textit{Condorcet-loser-dominates (CL-dominates)} another scoring rule $g$ if the set of profiles where $f$ selects a Condorcet loser is a proper subset of the set where $g$ does.
Since the Borda rule is the unique scoring rule that never selects a Condorcet loser \citep{Fishburn1976Borda,okamoto2019borda}, it CL-dominates all other scoring rules.
Our main theorem shows that no dominance relationship exists between any pair of non-Borda scoring rules.
In particular, any scoring rule arbitrarily close to the Borda rule (e.g., the rule with $s_2 = 0.499$ in the three-alternative case) is not better in terms of CL dominance than one obviously farther from the Borda rule (e.g., the rule with $s_2 = 0.01$ or $s_2 = 0$, namely, the plurality rule).
Our result further highlights the special status of the Borda rule among scoring rules: it is the only scoring rule that CL-dominates at least one scoring rule.

The concept of CL dominance was first introduced by \citet{doi2024condorcet}. While they restricted their attention to the comparison between the plurality rule and other scoring rules within the three-alternative case, we generalize their analysis to arbitrary pairs of scoring rules with any number of alternatives.

We believe that our study also makes a technical contribution. 
We prove the theorem by showing that for any non-Borda scoring rule $f$ and any scoring rule $f'$, there exists a preference profile in which $f$ selects the Condorcet loser and $f'$ does not (Proposition \ref{prop:existence}). 
Our proof of this proposition relies on induction on the number of alternatives. 
Specifically, we first show that, in the three-alternative case, there always exists a desired profile. 
Then we construct desired profiles for the general $m$-alternative case: 
for most pairs of scoring rules, a desired $m$-alternative profile can be obtained by inserting new alternatives into a three-alternative or $(m-1)$-alternative  profile and adding a ``dummy'' population of voters to it.
This method would be useful for extending properties of scoring rules obtained in cases with few alternatives to the general case.

Our study contributes to two strands of literature on scoring rules.
The first strand examines the performance of all scoring rules in terms of the frequency of selecting a Condorcet loser. For example, \citet{lepelley2000scoring} calculated the probability that each scoring rule selects a Condorcet loser in the three-alternative case under various distributional assumptions. Furthermore, \citet{kamwa2017scoring} calculated the probability that each scoring rule selects a Condorcet loser in the three-alternative case under specific preference restrictions. While these studies found that rules closer to the Borda rule perform better \textit{on average}, we show that no scoring rule that is not the Borda rule is better than any other scoring rule for every preference profile. In contrast to their analyses, our result holds regardless of the number of alternatives.

The second strand studies the characterization of the Borda rule.
\citet{smith1973aggregation} stated that the Borda rule is the only scoring rule under which a Condorcet winner never finishes last.
\citet{young1974axiomatization} showed that the Borda rule is the unique voting rule satisfying four natural properties.
\citet{Fishburn1976Borda} stated that the Borda rule is the unique scoring rule that never chooses a Condorcet loser.
\citet{saari1989dictionary,saari1990borda} proved that the Borda rule admits the fewest inconsistent choices across different sets of alternatives among the class of scoring rules.
\citet{maskin2025borda} characterized the Borda rule on the basis of a weaker version of \citeauthor{Arrow1951}'s (\citeyear{Arrow1951}) independence condition.
In contrast, this paper characterizes the Borda rule using a dominance relation that compares scoring rules.

Finally, we remark on studies of dominance relations among rules.
In matching theory, \citet{gale1962college} 
established that the deferred acceptance algorithm is the optimal rule among those that yield stable matchings, with respect to proposer efficiency.
More recently, \citet{abdulkadiroglu2020efficiency} and \citet{dougan2022robust} showed that the top trading cycle mechanism is not dominated by any other efficient and strategy-proof mechanism in terms of stability. 
In summary, several well-known rules emerge as maximal elements with respect to some dominance relation.
Following this approach, we establish the relative advantage of the Borda rule in the context of voting. 
In voting theory, \citet{maskin1995book} compared the performance of voting rules within a class of rules that satisfy reasonable properties. 
Maskin showed that the pairwise majority rule is the unique voting rule that is transitive on the widest class of domains of preference profiles. 
While Maskin highlights the special status of the pairwise majority rule with respect to transitivity, we show that the Borda rule is the unique rule that dominates some scoring rule under another criterion, CL dominance. 

The remainder of this paper is organized as follows.
Section \ref{sec:model} introduces the model and basic definitions.
Section \ref{sec:result} presents our main result.
Section \ref{sec:conc} concludes.
All proofs are provided in the Appendix.

\section{Model}
\label{sec:model}

Let $X=\{x_1,\dots,x_m\}$ be a set of alternatives with $m\ge 3$.
Let $\mathcal{I}$ be a countably infinite set of potential voters,
and $\mathcal{N}$ be the set of finite subsets of $\mathcal{I}$.
For each $i\in \mathcal{I}$, voter $i$'s \textbf{\textit{preference relation}} $\succsim_i$ is a linear order on $X$.%
\footnote{A linear order on $X$ is a complete, transitive, and anti-symmetric binary relation on $X$.}
For convenience, we define the ranking of each alternative under a preference relation $\succsim_i$ by $$r(x,\succsim_i)=|\{y\in X:y\succsim_i x\}|\in\{1,\dots,m\}.$$
Let $\mathscr{P}$ be the set of all linear orders on $X$.
A \textbf{\textit{preference profile}} is a profile $\succsim = (\succsim_i )_{i\in N} \in \mathscr{P}^N$ such that $N \in \mathcal{N}$.
For notational convenience, we denote the sets of voters associated with profiles $\succsim$, $\succsim'$, and $\succsim^*$ by $N$, $N'$, and $N^*$, respectively.
We use similar notation for other preference profiles unless there is a risk of confusion.
The set of preference profiles is denoted by $\mathcal{D}$, that is, $\mathcal{D}=\bigcup_{N\in \mathcal{N}} \mathscr{P}^N$.

A \textbf{\textit{score vector}} is an $m$-dimensional vector $s=(s(1),\dots,s(m))\in [0,1]^m$ satisfying $$1=s(1)\ge\cdots\ge s(m)=0.$$
For example, the \textit{\textbf{Borda score vector}} $s^B$ is such that $s^B(k)-s^B(k+1)=s^B(k+1)-s^B(k+2)$ for all $k\in\{1,\dots,m-2\}$.
A \textbf{\textit{scoring rule}} with a score vector $s$ is the function $f^s:\mathcal{D}\to 2^X\backslash\{\emptyset\}$ such that for each $\succsim\in\mathcal{D}$, $$f^s(\succsim)=\argmax_{x\in X} \sum_{i\in N}s(r(x,\succsim_i)).$$
The \textit{\textbf{Borda rule}} $f^B$ is the scoring rule with $s^B$.
We often identify a scoring rule with its associated score vector.
An alternative $x$ is a \textbf{\textit{Condorcet loser}} at $\succsim$ if for each $y\neq x$, $$|\{i\in N:y\succsim_i x\}|>|\{i\in N:x\succsim_i y\}|.$$
For each scoring rule $f$, let $$L(f)\equiv \{ ~ \succsim\in \mathcal{D}:f(\succsim)\text{ contains a Condorcet loser at }\succsim ~\}.$$
We say that a scoring rule $f$ \textbf{\textit{CL-dominates}} another scoring rule $f'$ if $$L(f)\subsetneq L(f').$$

\section{Main results}
\label{sec:result}

We begin by clarifying the relationship between the Borda rule and other scoring rules in terms of the Condorcet loser criterion.
As established by the following celebrated theorem, the Borda rule is the unique scoring rule that never selects a Condorcet loser.

\begin{theorem}[Fishburn and Gehrlein, 1976; Okamoto and Sakai, 2019]\label{thm:FG}
  For any non-Borda scoring rule $f \neq f^B$, $L(f) \neq \emptyset$. Furthermore, $L (f^B) = \emptyset$.
\end{theorem}

A direct consequence of Theorem \ref{thm:FG} is that the Borda rule CL-dominates all other scoring rules.
The remaining question concerns the dominance relationship among non-Borda scoring rules.
As mentioned in the Introduction, probabilistic analyses by \citet{lepelley2000scoring} suggest that rules closer to the Borda rule perform better in avoiding the Condorcet loser.
From this observation, one might expect a dominance relationship where rules closer to Borda would dominate those farther away.
However, this is not the case:

\begin{proposition}
  \label{prop:existence}
  For any non-Borda scoring rule $f\neq f^B$ and another scoring rule $f'$, there exists a preference profile $\succsim\in\mathcal{D}$ such that
  \begin{itemize}
    \item[(1)] $x_1$ is a Condorcet loser at $\succsim$,
    \item[(2)] $f(\succsim)=\{x_1\}$, and
    \item[(3)] $x_1\notin f'(\succsim).$
  \end{itemize}
\end{proposition}

By Theorem \ref{thm:FG} and Proposition \ref{prop:existence}, we obtain our main result.

\begin{theorem}
  \label{thm:main}
  Take any non-Borda scoring rule $f'$. A scoring rule $f$ CL-dominates $f'$ if and only if $f$ is the Borda rule $f^B$.
\end{theorem}

This result further highlights the special status of the Borda rule: it is the unique scoring rule that CL-dominates some scoring rule.
In other words, no hierarchy exists among non-Borda scoring rules in terms of CL dominance; for any pair of distinct non-Borda rules, there is always a situation where one fails to avoid the Condorcet loser while the other succeeds.

In the Appendix, we prove Proposition \ref{prop:existence} by constructing a desired preference profile for every pair $(f,f')$ satisfying the condition. 
Since these profiles heavily rely on the score vectors, it is difficult to offer a simple construction that can be used for every case. 
To solve this problem, we prove the statement
by induction on the number of alternatives.
First, we construct desired preference profiles for the three-alternative case, depending on the score vectors.
Then, to extend the result in the three-alternative case to the $m$-alternative case, we show that for most pairs of score vectors, we can embed $(m-1)$-dimensional or three-dimensional score vectors (Lemma \ref{lem:convolution} in Appendix).
Finally, we offer a method for constructing the desired profile from such a profile with fewer alternatives.
In what follows, we explain the intuition behind our proof of Proposition \ref{prop:existence}, focusing on two cases: the three- and four-alternative cases.

\paragraph{The three-alternative case.}
In this case, we can identify each score vector $s = (s(1), s(2), s(3))$ with $s(2)$ since we assume that $s(1) = 1$ and $s(3) = 0$. The Borda rule corresponds to $s(2) = {1\over 2}$, and the plurality rule corresponds to $s(2) = 0$.
Suppose that $f$ and $f'$ are associated with score vectors $s$ and $s'$, respectively.

\begin{table}[h]
  \centering
  \caption{Preference profile $\succsim$ with $(2L + L')$ voters and 3 alternatives ($x$, $y$ and $z$)}
  \label{tab:threeex}
  \begin{tabular}{|c|c|c|} \hline
    $L$ voters & $L$ voters & $L'$ voters \\ \hline
    $y$ & $z$ & $y$ \\ \hline
    $x$ & $x$ & $z$ \\ \hline
    $z$ & $y$ & $x$ \\ \hline
  \end{tabular}
\end{table}

For instance, consider the case with ${1\over 2} < s'(2) < s(2) \leq 1$ (cf.~Appendix~A.1.1).
Then, both rules assign a higher weight to the second rank than the Borda rule, and $f$ assigns a higher weight to the second rank than $f'$.
We construct a preference profile $\succsim$ in which a Condorcet loser $x$ at $\succsim$ is chosen under $f$ but not chosen under $f'$.
Consider the preference profile in Table \ref{tab:threeex}.
In the left two columns, $x$ and $y$ are tied under the pairwise majority comparison.
Due to the voters in the right column, $x$ is beaten by $y$.
Similarly, we can see that $x$ is also beaten by $z$ under the pairwise majority comparison.
That is, $x$ is a Condorcet loser.
Since $x$ is in the second place for $2L$ voters in the left two columns, we can make $f$ choose $x$ by letting $2L$ be sufficiently larger than $L'$.
However, if we set $2L$ to be too large relative to $L'$, then $f'$ would also choose $x$.
By balancing these two values, we can construct a preference profile in which $x$ is not chosen by $f'$.
It is possible because $s'(2) < s(2)$ implies that the threshold for $f$ is lower than that of $f'$.

In Appendix~A.1, we explicitly construct desired preference profiles for the other cases.

\paragraph{The four-alternative case.}
By utilizing the preference profiles obtained in the three-alternative case, we construct a desired preference profile.
Let $f$ and $f'$ be two scoring rules associated with $s$ and $s'$, respectively.
Since $f$ is non-Borda, at least one of $(s(1), s(2), s(3))$, $(s(1), {s(2)+s(3)\over 2}, s(4))$, and $(s(2), s(3), s(4))$ is also a non-Borda scoring rule in the three-alternative case.
Assume that the following conditions hold:
(i) $(s(1), s(2), s(3))$ is a score vector (i.e., $s(1) > s(3)$);
(ii) $(s'(1), s'(2), s'(3))$ is a score vector (i.e., $s'(1) > s'(3)$);
(iii) $(s(1), s(2), s(3))$ is not the Borda rule;
and (iv) $(s(1), s(2), s(3)) \neq (s'(1), s'(2), s'(3))$.
Under these conditions, we can construct a desired preference profile from the profile in the case where $m=3$ as follows (note that we discuss later the case where the conditions do not hold).

By the result for the three-alternative case, there exists a preference profile $\succsim$ over $\{ x,y,z \}$ such that a Condorcet loser at $\succsim$ is chosen under $(s(1), s(2), s(3))$ but not chosen under $(s'(1), s'(2), s'(3))$.
Let $\succsim^1$ be the preference profile over $\{ x,y,z,w \}$ obtained by adding the new alternative $w$ to the fourth place for all voters.
By construction, $x$ is chosen under $f$ and not chosen under $f'$ in this profile since $w$ is never chosen.
Similarly, it is easy to verify that $x$ is beaten by $y$ and $z$ under the pairwise majority comparison; however, $w$ is beaten by $x$ since $w$ is in the fourth place for all voters.
We adjust the preference profile so that $x$ is beaten by $w$ under the pairwise majority comparison while preserving the other properties.

\begin{table}[h]
  \centering
  \caption{``Uniform'' preference profile}
  \label{tab:fourex}
  \begin{tabular}{|c|c|c|c|c|c|} \hline
    $L$ voters & $L$ voters & $L$ voters & $L$ voters & $L$ voters & $L$ voters \\ \hline
    $x$ & $x$ & $y$ & $y$ & $z$ & $z$  \\ \hline
    $y$ & $z$ & $x$ & $z$ & $x$ & $y$ \\ \hline
    $z$ & $y$ & $z$ & $x$ & $y$ & $x$  \\ \hline
  \end{tabular}
\end{table}

To achieve this, consider the ``uniform'' preference profile over $\{ x,y,z\}$ in Table \ref{tab:fourex}.
Let $\succsim^2$ be the preference profile over $\{ x,y,z,w\}$ constructed from the uniform preference profile by inserting $w$ in the third place.
Similarly, we define $\succsim^3$ and $\succsim^4$ as the preference profiles constructed from the uniform preference profile by inserting $w$ in the second and first places, respectively.
We assume that the numbers of voters at $\succsim^1$, $\succsim^2$, $\succsim^3$, and $\succsim^4$ are equal.

Note that for all voters, $x$ is ranked above $w$ at $\succsim^1$ and below $w$ at $\succsim^4$. Moreover, by the symmetry of $\succsim^2$ and $\succsim^3$, the total number of wins of $x$ against $w$ at $\succsim^2$ and $\succsim^3$ is equal to the number of wins of $w$ against $x$.
Therefore, in the merged preference profile $(\succsim^1,\succsim^2, \succsim^3, \succsim^4)$, $x$ and $w$ are tied under the majority comparison.
Define the preference profile $\succsim^*$ by $(\succsim^1,\succsim^2, \succsim^3, \succsim^4, \succsim^+)$, where $\succsim^+$ consists of one individual who prefers $w$ to $z$ to $y$ to $x$.
By construction, $x$ is beaten by $w$ at $\succsim^*$ under the majority comparison.
Since $x$ is beaten by $y$ and $z$ at $(\succsim^1, \succsim^+)$ and is tied with them under $(\succsim^2, \succsim^3, \succsim^4)$, we can see that $x$ is also beaten by $y$ and $z$ at $\succsim^*$; that is, $x$ is a Condorcet loser.

Finally, we check which alternatives $f$ and $f'$ choose.
By letting the number of voters in $(\succsim^1,\succsim^2, \succsim^3, \succsim^4)$ be sufficiently large, we can ignore the score from the one-voter preference profile $\succsim^+$.
Note that $x$, $y$, and $z$ get the same score at $(\succsim^2, \succsim^3, \succsim^4)$ under each score vector.
Since $(s(1),s(2),s(3))$ chooses $x$ at $\succsim$, $x$ obtains a higher score than $y$ and $z$ at $\succsim^1$.
Hence, at $(\succsim^1, \succsim^2, \succsim^3, \succsim^4)$, $f$ assigns a higher score to $x$ than those of $y$ and $z$.
Furthermore, by construction, $w$ gets the average score at $(\succsim^1, \succsim^2, \succsim^3, \succsim^4)$.
Since the average score is the highest value only when all the alternatives get the same score, $x$ obtains a higher score than $w$ at $(\succsim^1, \succsim^2, \succsim^3, \succsim^4)$.
Therefore, $x$ is chosen by $f$ at $(\succsim^1, \succsim^2, \succsim^3, \succsim^4)$.
By a similar argument, we can see that $x$ is not chosen by $f'$ at $(\succsim^1, \succsim^2, \succsim^3, \succsim^4)$ since $(s'(1),s'(2),s'(3))$ does not choose $x$ at $\succsim$.

Recall that, in the above construction, we assumed the four conditions (i)--(iv) for the vector $(s(1),s(2),s(3))$ (and $ (s'(1),s'(2),s'(3))$).
Lemma \ref{lem:convolution} in Appendix shows that for any pair $(s, s')$ except for the case where $s(2) = s(3)$ and $s'(2) = s'(3)$, if the conditions (i)--(iv) fail to hold for $(s(1),s(2),s(3))$, then the corresponding conditions for $(s(1), {s(2)+s(3)\over 2}, s(4))$ or $(s(2),s(3),s(4))$ hold.
Even in these cases, we can construct a desired profile from some profile in the case of fewer alternatives, as illustrated above.
We deal with the exceptional case where $s(2) = s(3)$ and $s'(2) = s'(3)$ separately by explicitly providing a preference profile.

\section{Conclusion}
\label{sec:conc}

In this paper, we compared scoring rules in terms of CL dominance.
Our main theorem shows that the Borda rule is the unique scoring rule that dominates at least one scoring rule.
This theoretical finding offers a new perspective in understanding Condorcetian properties of scoring rules.
Furthermore, the proof technique developed in this paper is also our contribution. We believe that in future studies, the method of embedding fewer-alternative cases into the general case can also be used when studying the properties of scoring rules, such as the performance with respect to the choice of Condorcet losers and winners and, more generally, the relationship between the pairwise majority rule and scoring rules. 

\section*{Appendix: Proof of Proposition \ref{prop:existence}}
\renewcommand{\thesubsection}{A.\arabic{subsection}}
\setcounter{subsection}{0}

Take any non-Borda scoring rule $f$ and another scoring rule $f'$. Let $s$ and $s'$ be the associated score vectors, respectively. We prove the proposition by mathematical induction.

\subsection{The case with $m=3$}

Let $X=\{x,y,z\}$ for readability. For any profile $\succsim$, let
\begin{align*}
  &n_1(\succsim)=|\{i\in N: x\succsim_i y\succsim_i z\}|, \\
  &n_2(\succsim)=|\{i\in N: x\succsim_i z\succsim_i y\}|,\\
  &n_3(\succsim)=|\{i\in N: y\succsim_i x\succsim_i z\}|,\\
  &n_4(\succsim)=|\{i\in N: y\succsim_i z\succsim_i x\}|,\\
  &n_5(\succsim)=|\{i\in N: z\succsim_i x\succsim_i y\}|,~\text{ and}\\
  &n_6(\succsim)=|\{i\in N: z\succsim_i y\succsim_i x\}|.
\end{align*}
Unless there is a risk of confusion, we just write $n_k(\succsim)$ as $n_k$ ($k=1,\dots,6$).
We consider six cases.

\subsubsection{The case with $\frac{1}{2}\leq s'(2)<s(2)\leq 1$}
The subcase with $s'(2)=\frac{1}{2}$ is obvious. Otherwise, let $b\in\mathbb{Q}_{++}$ be such that
\[
  \frac{1}{2s(2)-1}<b<\frac{1}{2s'(2)-1}.
\]
Take $p,q\in\mathbb{N}$ with $b=\frac{p}{q}$. Consider a profile $\succsim$ such that
\begin{align*}
  (n_1,n_2,n_3,n_4,n_5,n_6)=(0,0,bq,q,bq,0).
\end{align*}
Since $|\{i\in N: x\succsim_i y\}|=bq<bq+q=|\{i\in N: y\succsim_i x\}|$ and $|\{i\in N: x\succsim_i z\}|=bq<bq+q=|\{i\in N: z\succsim_i x\}|$ hold, $x$ is a Condorcet loser at $\succsim$.

By
\begin{align*}
  &\sum_{i\in N}s(r(x,\succsim_i))=1\cdot 0+s(2)\cdot 2bq = 2bqs(2),\\
  &\sum_{i\in N}s(r(y,\succsim_i))=1\cdot (bq+q)+s(2)\cdot 0 = bq + q,~\text{ and}\\
  &\sum_{i\in N}s(r(z,\succsim_i))=1\cdot bq+s(2)\cdot q = bq + qs(2) ,
\end{align*}
we have $ \sum_{i\in N}s(r(x,\succsim_i))-\sum_{i\in N}s(r(y,\succsim_i)) = 2bqs(2)-bq-q = q(b(2s(2)-1)-1) > 0$.
In addition, since $\sum_{i\in N}s(r(y,\succsim_i))\ge\sum_{i\in N}s(r(z,\succsim_i))$, we have $\sum_{i\in N}s(r(x,\succsim_i))>\sum_{i\in N}s(r(z,\succsim_i))$.
Thus, $f(\succsim)=\{x\}$.

By
\begin{align*}
  &\sum_{i\in N}s'(r(x,\succsim_i))=1\cdot 0+s'(2)\cdot 2bq=2bqs'(2) ~\text{ and}\\
  &\sum_{i\in N}s'(r(y,\succsim_i))=1\cdot (bq+q)+s'(2)\cdot 0=bq+q,
\end{align*}
we have $\sum_{i\in N}s'(r(y,\succsim_i))-\sum_{i\in N}s'(r(x,\succsim_i)) =bq+q-2bqs'(2) = q(1-b(2s'(2)-1)) > 0$.
Hence, $x\notin f'(\succsim)$.

\subsubsection{The case with $\frac{1}{2}< s(2)<s'(2)\leq 1$}
If $s'(2)=1$, let $b\in\mathbb{Q}_{++}$ be such that
\[
  \frac{s(2)}{1-s(2)}<b.
\]
Otherwise, let $b\in\mathbb{Q}_{++}$ be such that
\[
  \frac{s(2)}{1-s(2)}<b<\frac{s'(2)}{1-s'(2)}.
\]
Take $p,q\in\mathbb{N}$ with $b=\frac{p}{q}$. Take $c\in\mathbb{N}$ with
\[
  \frac{bs(2)+1}{2s(2)-1}<c.
\]
Consider a profile $\succsim$ such that
\[
  (n_1,n_2,n_3,n_4,n_5,n_6)=(cq,bq,cq,0,cq,bq+cq+q).
\]
Since $|\{i\in N: x\succsim_i y\}|=bq+2cq<bq+2cq+q=|\{i\in N: y\succsim_i x\}|$ and $|\{i\in N: x\succsim_i z\}|=bq+2cq<bq+2cq+q=|\{i\in N: z\succsim_i x\}|$ hold,
$x$ is a Condorcet loser at $\succsim$.

By
\begin{align*}
  &\sum_{i\in N}s(r(x,\succsim_i))=1\cdot (bq+cq)+s(2)\cdot 2cq = bq + cq + 2cq s(2),\\
  &\sum_{i\in N}s(r(y,\succsim_i))=1\cdot cq+s(2)\cdot (bq+2cq+q) = cq + bqs(2) + 2cqs(2)+q s(2),~\text{ and}\\
  &\sum_{i\in N}s(r(z,\succsim_i))=1\cdot (bq+2cq+q)+s(2)\cdot bq = bq+2cq+q + bq s(2),
\end{align*}
we have $\sum_{i\in N}s(r(x,\succsim_i))-\sum_{i\in N}s(r(y,\succsim_i)) =bq-s(2)(bq+q) = q(b(1-s(2))-s(2)) > 0$ and
$\sum_{i\in N}s(r(x,\succsim_i))-\sum_{i\in N}s(r(z,\succsim_i)) = -cq-q+s(2)(2cq-bq) = q(c(2s(2)-1)-(bs(2)+1)) > 0$.
Thus, $f(\succsim)=\{x\}$.

By
\begin{align*}
  &\sum_{i\in N}s'(r(x,\succsim_i))=1\cdot (bq+cq)+s'(2)\cdot 2cq = bq + cq + 2cq s'(2)~\text{ and}\\
  &\sum_{i\in N}s'(r(y,\succsim_i))=1\cdot cq+s'(2)\cdot (bq+2cq+q)  = cq + bqs'(2) + 2cqs'(2)+q s'(2),
\end{align*}
we have $ \sum_{i\in N}s'(r(y,\succsim_i))-\sum_{i\in N}s'(r(x,\succsim_i)) = -bq+s'(2)(bq+q ) = q(-b(1-s'(2))+s'(2)) > 0$.
Hence, $x\notin f'(\succsim)$.

\subsubsection{The case with $ 0\leq s(2) <\frac{1}{2}\leq s'(2)\leq 1$}
Let $b\in\mathbb{N}$ be such that
\[
  \frac{1+s(2)}{1-2s(2)}<b.
\]
Consider a profile $\succsim$ such that
\[
  (n_1,n_2,n_3,n_4,n_5,n_6)=(0,b+1,0,b,0,2).
\]
Since $|\{i\in N: x\succsim_i y\}|=b+1<b+2=|\{i\in N: y\succsim_i x\}|$ and $|\{i\in N: x\succsim_i z\}|=b+1<b+2=|\{i\in N: z\succsim_i x\}|$ hold,
$x$ is a Condorcet loser at $\succsim$.

By
\begin{align*}
  &\sum_{i\in N}s(r(x,\succsim_i))=1\cdot (b+1)+s(2)\cdot 0 = b+1,\\
  &\sum_{i\in N}s(r(y,\succsim_i))=1\cdot b+s(2)\cdot 2 = b + 2s(2),~\text{ and}\\
  &\sum_{i\in N}s(r(z,\succsim_i))=1\cdot 2+s(2)\cdot (2b+1) = 2 + 2bs(2) + s(2),
\end{align*}
we have $\sum_{i\in N}s(r(x,\succsim_i))-\sum_{i\in N}s(r(y,\succsim_i)) =1-2s(2) > 0$
and $\sum_{i\in N}s(r(x,\succsim_i))-\sum_{i\in N}s(r(z,\succsim_i))=b-1- 2bs(2) - s(2)= b(1-2s(2))-(1+s(2)) > 0$. Thus, $f(\succsim)=\{x\}$.

By
\begin{align*}
  &\sum_{i\in N}s'(r(x,\succsim_i))=1\cdot (b+1)+s'(2)\cdot 0 = b+ 1~\text{ and}\\
  &\sum_{i\in N}s'(r(z,\succsim_i))=1\cdot 2+s'(2)\cdot (2b+1)= 2 + 2bs'(2) + s'(2),
\end{align*}
we have $\sum_{i\in N}s'(r(z,\succsim_i))-\sum_{i\in N}s'(r(x,\succsim_i)) =1-b+s'(2)(2b+1) \ge 1-b+\frac{1}{2}\cdot(2b+1) > 0$.
Hence, $x\notin f'(\succsim)$.

\subsubsection{The case with $ 0\leq s'(2) \le\frac{1}{2}< s(2)\leq 1$}
Let $b\in\mathbb{N}$ be such that
\[
  \frac{1}{2s(2)-1}<b.
\]
Consider a profile $\succsim$ such that
\[
  (n_1,n_2,n_3,n_4,n_5,n_6)=(0,0,b,0,b,1).
\]
Since $|\{i\in N: x\succsim_i y\}|=b<b+1=|\{i\in N: y\succsim_i x\}|$ and $|\{i\in N: x\succsim_i z\}|=b<b+1=|\{i\in N: z\succsim_i x\}|$ hold,
$x$ is a Condorcet loser at $\succsim$.

By
\begin{align*}
  &\sum_{i\in N}s(r(x,\succsim_i))=1\cdot 0+s(2)\cdot 2b = 2b s(2) ,\\
  &\sum_{i\in N}s(r(y,\succsim_i))=1\cdot b+s(2)\cdot 1 = b + s(2),~\text{ and}\\
  &\sum_{i\in N}s(r(z,\succsim_i))=1\cdot (b+1)+s(2)\cdot 0 = b+1,
\end{align*}
we have $\sum_{i\in N}s(r(x,\succsim_i))-\sum_{i\in N}s(r(y,\succsim_i)) = 2b s(2) - b -s(2)  =b(2s(2)-1)-s(2)>1-s(2) \geq 0$
and
$\sum_{i\in N}s(r(x,\succsim_i))-\sum_{i\in N}s(r(z,\succsim_i)) = 2bs(2)-b-1=  b(2s(2)-1)-1 > 0$.
Hence, $f(\succsim)=\{x\}$.

By
\begin{align*}
  &\sum_{i\in N}s'(r(x,\succsim_i))=1\cdot 0+s'(2)\cdot 2b =  2b s'(2)~\text{ and}\\
  &\sum_{i\in N}s'(r(z,\succsim_i))=1\cdot (b+1)+s'(2)\cdot 0= b+1,
\end{align*}
we have $\sum_{i\in N}s'(r(z,\succsim_i))-\sum_{i\in N}s'(r(x,\succsim_i))=b+1- 2b s'(2) \ge b+1-\frac{1}{2}\cdot 2b > 0$.
Thus, $x\notin f'(\succsim)$.

\subsubsection{The case with $ 0\leq s(2)<s'(2) \le \frac{1}{2}$}
The subcase with $s'(2)=\frac{1}{2}$ is obvious. Otherwise, let $b\in\mathbb{Q}_{++}$ be such that
\[
  \frac{1}{2-4s(2)}<b<\frac{1}{2-4s'(2)}.
\]
Take $p,q\in\mathbb{N}$ with $b=\frac{p}{q}$. Consider a profile $\succsim$ such that
\[
  (n_1,n_2,n_3,n_4,n_5,n_6)=(bq,2bq,0,2bq,0,bq+q).
\]
Since $|\{i\in N: x\succsim_i y\}|=3bq<3bq+q=|\{i\in N: y\succsim_i x\}|$ and
$|\{i\in N: x\succsim_i z\}|=3bq<3bq+q=|\{i\in N: z\succsim_i x\}|$ hold,
$x$ is a Condorcet loser at $\succsim$.

By
\begin{align*}
  &\sum_{i\in N}s(r(x,\succsim_i))=1\cdot 3bq+s(2)\cdot 0 = 3bq,\\
  &\sum_{i\in N}s(r(y,\succsim_i))=1\cdot 2bq+s(2)\cdot (2bq+q) = 2bq + 2bqs(2) + qs(2),~\text{ and}\\
  &\sum_{i\in N}s(r(z,\succsim_i))=1\cdot (bq+q)+s(2)\cdot 4bq = bq+q+ 4bqs(2),
\end{align*}
we have $ \sum_{i\in N}s(r(x,\succsim_i))-\sum_{i\in N}s(r(y,\succsim_i)) = bq-2bqs(2) - qs(2) = q(b(1-2s(2))-s(2)) > q(\frac{1}{2}-s(2)) > 0$ and
$\sum_{i\in N}s(r(x,\succsim_i))-\sum_{i\in N}s(r(z,\succsim_i)) = 2bq-q- 4bq s(2) = q(b(2-4s(2))-1) > 0$.
Hence, $f(\succsim)=\{x\}$.

By
\begin{align*}
  &\sum_{i\in N}s'(r(x,\succsim_i))=1\cdot 3bq+s'(2)\cdot 0 = 3bq~\text{ and}\\
  &\sum_{i\in N}s'(r(z,\succsim_i))=1\cdot (bq+q)+s'(2)\cdot 4bq =  bq+q+ 4bqs'(2),
\end{align*}
we have $ \sum_{i\in N}s'(r(z,\succsim_i))-\sum_{i\in N}s'(r(x,\succsim_i)) =q-2bq+ 4bq s'(2) = q(1-b(2-4s'(2))) > 0$.
Thus,  $x\notin f'(\succsim)$.

\subsubsection{The case with $ 0\leq s'(2)<s(2) < \frac{1}{2}$}
Let $b\in\mathbb{Q}_{++}$ be such that
\[
  \frac{2-3s'(2)}{1-2s'(2)}<b<\frac{2-3s(2)}{1-2s(2)}.
\]
Take $p,q\in\mathbb{N}$ with $b=\frac{p}{q}$. Take $c\in \mathbb{N}$ with
\[
  b+\frac{1+2s(2)}{1-2s(2)}<c.
\]
Consider a profile $\succsim$ such that
\[
  (n_1,n_2,n_3,n_4,n_5,n_6)=(0,cq+2q,bq,cq,bq,3q).
\]
Since $|\{i\in N: x\succsim_i y\}|=bq+cq+2q<bq+cq+3q=|\{i\in N: y\succsim_i x\}|$ and $|\{i\in N: x\succsim_i z\}|=bq+cq+2q<bq+cq+3q=|\{i\in N: z\succsim_i x\}|$ hold, $x$ is a Condorcet loser at $
\succsim$.

By
\begin{align*}
  &\sum_{i\in N}s(r(x,\succsim_i))=1\cdot (cq+2q)+s(2)\cdot 2bq = cq+2q+ 2bq s(2),\\
  &\sum_{i\in N}s(r(y,\succsim_i))=1\cdot (bq+cq)+s(2)\cdot 3q = bq + cq + 3q s(2),~\text{ and}\\
  &\sum_{i\in N}s(r(z,\succsim_i))=1\cdot (bq+3q)+s(2)\cdot (2cq+2q) = bq+3q+ 2cqs(2)+2q s(2),
\end{align*}
we have $\sum_{i\in N}s(r(x,\succsim_i))-\sum_{i\in N}s(r(y,\succsim_i)) = 2q+ 2bq s(2) - bq - 3q s(2) = q(2-3s(2)-b(1-2s(2))) > 0$ and
$\sum_{i\in N}s(r(x,\succsim_i))-\sum_{i\in N}s(r(z,\succsim_i))= cq-q+ 2bq s(2) -  bq- 2cqs(2)-2q s(2)= q((c-b)(1-2s(2))-(1+2s(2))) > 0$.
Hence, $f(\succsim)=\{x\}$.

By
\begin{align*}
  &\sum_{i\in N}s'(r(x,\succsim_i))=1\cdot (cq+2q)+s'(2)\cdot 2bq = cq+2q+ 2bq s'(2)~\text{ and}\\
  &\sum_{i\in N}s'(r(y,\succsim_i))=1\cdot (bq+cq)+s'(2)\cdot 3q = bq + cq + 3q s'(2),
\end{align*}
we have $\sum_{i\in N}s'(r(y,\succsim_i))-\sum_{i\in N}s'(r(x,\succsim_i))=bq + 3q s'(2) -2q - 2bq s'(2)= q(b(1-2s'(2))-(2-3s'(2))) > 0$.
Thus, $x\notin f'(\succsim)$.

\subsection{The general case}

Suppose $m>3$ and the statement holds in the case where the number of alternatives is less than $m$.
For each $t\in\{1,m\}$, let $s_{-t}$ (resp. $s'_{-t}$) be the $(m-1)$-dimensional vector obtained by removing the $t$-th element from $s$ (resp. $s'$), and normalizing them if we can. Formally,
\begin{align*}
  &s_{-1}=
  \begin{cases}
    (0,\dots,0) &\text{ if }  s(2)= 0,\\
    (\frac{s(2)}{s(2)},\frac{s(3)}{s(2)},\dots,\frac{s(m)}{s(2)}) &\text{ if } s(2)\neq 0,
  \end{cases} ~~~ \text{and}\\
  &s_{-m}=
  \begin{cases}
    (1,\dots,1) &\text{ if }  s(m-1)= 1,\\
    (\frac{s(1)-s(m-1)}{s(1)-s(m-1)},\frac{s(2)-s(m-1)}{s(1)-s(m-1)},\dots,\frac{s(m-1)-s(m-1)}{s(1)-s(m-1)}) &\text{ if } s(m-1)\neq 1.
  \end{cases}
\end{align*}
In addition, let $s_{\text{ave}}$ be the three-dimensional vector such that $s_{\text{ave}}(1)=s(1)$, $s_{\text{ave}}(2)=\frac{1}{m-2}\sum_{k=2}^{m-1} s(k)$, and $s_{\text{ave}}(3)=s(m)$. Similarly, we also define $s'_{\text{ave}}$.

\begin{lemma}
  \label{lem:convolution}
  Suppose $(s,s')$ does not satisfy the following condition: $s(2)=s(m-1)=\frac{1}{2}$ and $s'(2)=s'(m-1)$. Then one of the following (i), (ii), or (iii) holds:
  \begin{itemize}
    \item[(i)] All of the following conditions 1(a)-1(d) hold:
      \begin{itemize}
        \item[1(a):] $s_{-1}$ is a score vector for $m-1$ alternatives.
        \item[1(b):] $s'_{-1}$ is a score vector for $m-1$ alternatives.
        \item[1(c):] $s_{-1}$ is not the Borda score vector for $m-1$ alternatives.
        \item[1(d):] $s_{-1}\neq s'_{-1}$.
      \end{itemize}
    \item[(ii)] All of the following conditions m(a)-m(d) hold:
      \begin{itemize}
        \item[m(a):] $s_{-m}$ is a score vector for $m-1$ alternatives.
        \item[m(b):] $s'_{-m}$ is a score vector for $m-1$ alternatives.
        \item[m(c):] $s_{-m}$ is not the Borda score vector for $m-1$ alternatives.
        \item[m(d):] $s_{-m}\neq s'_{-m}$.
      \end{itemize}
    \item[(iii)] All of the following conditions ave(a)-ave(d) hold:
      \begin{itemize}
        \item[ave(a):] $s_{\text{ave}}$ is a score vector for $3$ alternatives.
        \item[ave(b):] $s'_{\text{ave}}$ is a score vector for $3$ alternatives.
        \item[ave(c):] $s_{\text{ave}}$ is not the Borda score vector for $3$ alternatives.
        \item[ave(d):] $s_{\text{ave}}\neq s'_{\text{ave}}$.
      \end{itemize}
  \end{itemize}
\end{lemma}

To show this lemma, we first show two claims.\medskip\\
\noindent\textbf{Claim 1.} All but at most one of 1(c), m(c) and ave(c) hold.

\begin{proof}
  Suppose, for contradiction, that at least two of 1(c), m(c), and ave(c) do not hold.\medskip\\
  \textbf{Case 1.} Consider the case where both 1(c) and m(c) do not hold. Since 1(c) does not hold, $s(2)-s(3)=s(3)-s(4)=\cdots=s(m-1)-s(m)$. In addition, since m(c) does not hold, $s(1)-s(2)=s(2)-s(3)=\cdots =s(m-2)-s(m-1)$. Thus we have $s(1)-s(2)=s(2)-s(3)=\cdots =s(m-1)-s(m)$. It contradicts the assumption that $f$ is not the Borda rule.\medskip\\
  \textbf{Case 2.} Consider the case where both 1(c) and ave(c) do not hold. Since 1(c) does not hold, there exists $\alpha\leq\frac{1}{m-2}$ such that $s(m-k)=k\alpha$ for all $k\in\{1,\dots,m-2\}$. Since ave(c) does not hold, $s_{\text{ave}}(2)=\frac{1}{m-2}\sum_{k=2}^{m-1} s(k)=\frac{1}{2}.$ Thus we have $\alpha=\frac{1}{m-1}$. It contradicts the assumption that $f$ is not the Borda rule.\medskip\\
  \textbf{Case 3.} Consider the case where both m(c) and ave(c) do not hold. Since m(c) does not hold, there exists $\alpha\leq\frac{1}{m-2}$ such that $s(k+1)=1-k\alpha$ for all $k\in\{1,\dots,m-2\}$. Since ave(c) does not hold, $s_{\text{ave}}(2)=\frac{1}{m-2}\sum_{k=2}^{m-1} s(k)=\frac{1}{2}.$ Thus we have $\alpha=\frac{1}{m-1}$. It contradicts the assumption that $f$ is not the Borda rule.
\end{proof}

\noindent\textbf{Claim 2.} Both 1(d) and m(d) do not hold if and only if $s(2)=s(m-1)\neq s'(2)=s'(m-1)$.

\begin{proof}
  The ``if'' direction is obvious. We show the ``only if'' direction. Suppose that both 1(d) and m(d) do not hold, which means $s_{-1} = s'_{-1}$ and $s_{-m} = s'_{-m}$. If $s(2) = 0$, then $s_{-1} = (0, \dots, 0)$, and hence $s'_{-1}=(0, \dots, 0)$. It contradicts $s\neq s'$. Thus we have $s(2) \neq 0$. Similarly, we have $s'(2)\neq0$, $s(m-1)\neq 1$, and $s'(m-1)\neq 1$. Since $s_{-1} = s'_{-1}$, we have $\frac{s(k)}{s(2)} = \frac{s'(k)}{s'(2)}$ for all $k \in \{2, \dots, m-1\}$. Let $\alpha = \frac{s(2)}{s'(2)} > 0$, giving $s(k) = \alpha s'(k)$.
  Since $s_{-m} = s'_{-m}$, we have $\frac{s(k) - s(m-1)}{1 - s(m-1)} = \frac{s'(k) - s'(m-1)}{1 - s'(m-1)}$ for all $k \in \{2, \dots, m-1\}$. Let $\beta = \frac{1 - s(m-1)}{1 - s'(m-1)} > 0$, giving $s(k) - s(m-1) = \beta (s'(k) - s'(m-1))$.
  Then we have
  \begin{equation*}
    \alpha s'(k) - s(m-1) = \beta s'(k) - \beta s'(m-1),
  \end{equation*}
  that is, $(\alpha - \beta) s'(k) = s(m-1) - \beta s'(m-1)$.
  This equation must hold for all $k \in \{2, \dots, m-1\}$. Since the right-hand side is constant with respect to $k$, the left-hand side must also be constant.
  If $\alpha = \beta$, we have $s(m-1) = \beta s'(m-1) = \alpha s'(m-1)$. Then, by $\alpha = \beta= \frac{1 - s(m-1)}{1 - s'(m-1)}$, we have $s(m-1)=s'(m-1)$.
  This implies $\alpha = 1$, leading to $s(k) = s'(k)$ for all $k$, which contradicts $s \neq s'$. Thus, $\alpha \neq \beta$.
  Since $\alpha \neq \beta$, we have $s'(k)$ must be a constant for all $k \in \{2, \dots, m-1\}$. Therefore, since $s\neq s'$, $s(2) = s(m-1) \neq s'(2) = s'(m-1)$.
\end{proof}

\begin{proof}[Proof of Lemma \ref{lem:convolution}]
  It is straightforward by definition that ave(a) and ave(b) hold. If both ave(c) and ave(d) hold, then we are done. Suppose not. \medskip\\
  \textbf{Case 1.} Consider the case where ave(c) does not hold. Then we have
  \begin{align}
    s_{\text{ave}}(2)=\frac{1}{m-2}\sum_{k=2}^{m-1} s(k)=\frac{1}{2}.\label{avec}
  \end{align}
  Therefore both 1(a) and m(a) hold. In addition, by Claim 1, both 1(c) and m(c) hold. If both 1(b) and 1(d) hold, then we are done. Suppose not.
  \begin{itemize}
    \item Suppose 1(b) does not hold. Then $s'(2)=\cdots=s'(m-1)=0$. Thus m(b) holds. We show, by contradiction, that m(d) holds. Suppose not. Then $s_{-m}=s'_{-m}$. By \eqref{avec}, we have $s(2)=\cdots=s(m-1)=\frac{1}{2}$. It contradicts the assumption of this lemma. Thus m(d) holds, and hence m(a)-m(d) hold.

    \item Suppose 1(d) does not hold. We show, by contradiction, that m(d) holds. Suppose not. Then Claim 2 and \eqref{avec} imply that $s(2)=\cdots=s(m-1)=\frac{1}{2}$ and $s'(2)=\cdots=s'(m-1)$. It contradicts the assumption of this lemma. In addition, we show, by contradiction, that m(b) holds. Suppose not. Then $s'(2)=\cdots=s'(m-1)=1$. Therefore, since 1(d) does not hold, we have $s(2)=\cdots=s(m-1)$. By \eqref{avec}, $s(2)=\cdots=s(m-1)=\frac{1}{2}$. It contradicts the assumption of this lemma. Thus m(b) holds, and hence m(a)-m(d) hold.
  \end{itemize}
  \textbf{Case 2.} Consider the case where ave(d) does not hold. Then we have
  \begin{align}
    \frac{1}{m-2}\sum_{k=2}^{m-1} s(k)=\frac{1}{m-2}\sum_{k=2}^{m-1} s'(k).\label{aved}
  \end{align}
  We show, by contradiction, that 1(a) holds. Suppose not. Then $s(2)=\cdots =s(m-1)=0$. By \eqref{aved}, $s'(2)=\cdots=s'(m-1)=0$. It contradicts $s\neq s'$. Similarly, 1(b), m(a), and m(b) hold. In addition, by Claim 1, either 1(c) or m(c) holds.
  \begin{itemize}
    \item Suppose 1(c) holds. We show, by contradiction, that 1(d) holds. Suppose not. Then for all $k\in \{2,\ldots, m-1\}$,
      \begin{align}
        \frac{s(k)}{s(2)}=\frac{s'(k)}{s'(2)}. \label{1d}
      \end{align}
      Together with \eqref{aved}, we have
      \begin{align}
        \frac{s(2)}{s'(2)}\sum_{k=2}^{m-1} s'(k)=\sum_{k=2}^{m-1} s'(k).
      \end{align}
      Since $\sum_{k=2}^{m-1} s'(k)\neq 0$, $s(2)=s'(2)$. Then \eqref{1d} implies that $s(k)=s'(k)$ for all $k\in\{2,\ldots,m-1\}$, and hence $s=s'$, which is a contradiction. Thus 1(a)-1(d) hold.

    \item Suppose m(c) holds. We show, by contradiction, that m(d) holds. Suppose not. Then for all $k\in\{2,\ldots,m-1\}$,
      \begin{align}
        \frac{s(k)-s(m-1)}{s(1)-s(m-1)}=\frac{s'(k)-s'(m-1)}{s'(1)-s'(m-1)}. \label{md}
      \end{align}
      Note that $s(1)=s'(1)=1$. Summing both sides over $k\in\{2,\ldots,m-1\}$, we obtain
      \begin{align*}
        \frac{\sum_{k=2}^{m-1}s(k) - (m-2)s(m-1)}{1-s(m-1)} = \frac{\sum_{k=2}^{m-1}s'(k) - (m-2)s'(m-1)}{1-s'(m-1)}.
      \end{align*}
      Together with \eqref{aved}, we have
      \begin{align*}
        \sum_{k=2}^{m-1}s(k)\{s(m-1) - s'(m-1)\} = (m-2)\{s(m-1) - s'(m-1)\}.
      \end{align*}
      Since m(a) holds, $s(m-1) \neq 1$. Therefore, $\sum_{k=2}^{m-1}s(k)\neq m-2$. Thus we have $s(m-1) = s'(m-1)$.
      Then, by \eqref{md}, $s(k) = s'(k)$ for all $k \in \{2, \ldots, m-1\}$. This contradicts $s \neq s'$. Thus, m(d) holds, and hence m(a)-m(d) hold. \qedhere
  \end{itemize}
\end{proof}
Now we consider four cases.

\subsubsection{The case where Lemma \ref{lem:convolution}(i) or (ii) holds}

We only deal with the case where (ii) holds, since the other case can be handled in a similar way. Then, by the assumption of induction, there exist a preference profile $\succsim^1\in\mathcal{D}$ such that
\begin{enumerate}[label=(\roman*)]
  \item $r(x_m,\succsim^1_i)=m$ for all $i\in N^1$,
  \item $|\{ i\in N^1:x_k\succsim^1_i x_1\}|>|\{i\in N^1:x_1\succsim^1_i x_k \}|$ for all $k\in \{2,3,\dots,m-1\}$,
  \item $\sum_{i\in N^1}s(r(x_1,\succsim^1_i))>\sum_{i\in N^1}s(r(x_k,\succsim^1_i))$ for all $k\in \{2,3,\dots,m-1\}$, and
  \item $\sum_{i\in N^1}s'(r(x_2,\succsim^1_i))>\sum_{i\in N^1}s'(r(x_1,\succsim^1_i))$.
\end{enumerate}
For each $\ell\in\{2,\dots,m\}$, let $\succsim^\ell\in\mathcal{D}$ be such that there exists some $t\in\mathbb{N}$ such that
\begin{itemize}
  \item for each linear order $\succsim^\bullet$ on $X\backslash\{x_m\}$, $|\{i\in N^\ell:\succsim^\ell_i|_{X\backslash\{x_m\}}=\succsim^\bullet\}|=t$, 
  \item $r(x_m,\succsim^\ell_i)=m-\ell+1$ for all $i\in N^\ell$, and
  \item $N^1, N^2, \ldots, N^\ell$ are mutually disjoint.
\end{itemize}
Then, for all $\ell\in\{2,\dots,m\}$,
\begin{enumerate}[label=(\roman*)']
  \item $| \{ i\in N^\ell:x_k\succsim^\ell_i x_1 \} |= | \{ i\in N^\ell:x_1\succsim^\ell_i x_k \}|$ for all $k\in \{2,3,\dots,m-1\}$,
  \item $\sum_{i\in N^\ell}s(r(x_1,\succsim^\ell_i))=\sum_{i\in N^\ell}s(r(x_k,\succsim^\ell_i))$ for all $k\in \{2,3,\dots,m-1\}$, and
  \item  $\sum_{i\in N^\ell}s'(r(x_2,\succsim^\ell_i))=\sum_{i\in N^\ell}s'(r(x_1,\succsim^\ell_i))$.
\end{enumerate}
In addition, consider a single-voter profile $\succsim^+ \in \mathscr{P}$ consisting of a voter $i^+\notin \bigcup_{\ell=1}^m N^\ell$ whose preference relation satisfies $x_m \succsim^+_{i^+} x_{m-1} \succsim^+_{i^+} \cdots \succsim^+_{i^+} x_2 \succsim^+_{i^+} x_1$.
By replicating the profiles $\succsim^1, \dots, \succsim^m$ if necessary, we assume without loss of generality that
\begin{itemize}
  \item $|N^1| = |N^2| = \cdots = |N^m| = K$ for some integer $K$ and
  \item for any $x_k \in X \setminus \{x_1\}$, if the total scores of $x_1$ and $x_k$ under $s$ (resp. $s'$) in $\bigcup_{\ell=1}^m N^\ell$ are not equal, the absolute difference between them is strictly greater than $1$.
\end{itemize}
This ensures that the addition of the single voter $i^+$ cannot overturn any strict score inequalities established among $\succsim^1,\dots,\succsim^m$.
We construct $\succsim$ by combining this profile and other profiles $\succsim^2,\succsim^3,\dots,\succsim^m$ and $\succsim^+$;
that is, we define $\succsim=(\succsim^1,\dots,\succsim^m,\succsim^+)\in\mathscr{P}^{mK+1}$.\medskip\\
\textbf{Claim 3.} $x_1$ is a Condorcet loser at $\succsim$.
\begin{proof}
  It is straightforward from the definition of $\succsim$ that, for all $k\in\{2,\dots,m-1\}$, $|\{i\in N:x_1\succsim_i x_k\}|<|\{i\in N:x_k\succsim_i x_1\}|$. We show that $|\{i\in N:x_1\succsim_i x_m\}|<|\{i\in N:x_m\succsim_i x_1\}|$. Since $r(x_m,\succsim^1_i)=m$ for all $i\in N^1$ and $r(x_m,\succsim^m_i)=1$ for all $i\in N^m$,
  \begin{align*}
    |\{ i\in N^1:x_m\succsim^1_i x_1 \}|+|\{ i\in N^m:x_m\succsim^m_i x_1 \}|
    =|\{ i\in N^1:x_1\succsim^1_i x_m \}|+|\{ i\in N^m:x_1\succsim^m_i x_m \}|
    =K.
  \end{align*}
  In addition, by the definitions of $\succsim^2,\dots, $ and $\succsim^{m-1}$, we have that for all $k\in\{2,\dots,m-1\}$,
  \begin{align*}
    &|\{ i\in N^k:x_m\succsim^k_i x_1 \}|+|\{i\in N^{m-k+1}:x_m\succsim^{m-k+1}_i x_1 \}|\\
    &=|\{ i\in N^k:x_1\succsim^k_i x_m \}|+|\{ i\in N^{m-k+1}:x_1\succsim^{m-k+1}_i x_m \}|\\
    &=K.
  \end{align*}
  Therefore, we have
  \begin{equation*}
    |\{i\in N: x_1\succsim_i x_m\}|=\frac{1}{2}mK<\frac{1}{2}mK+1=|\{i\in N: x_m\succsim_i x_1\}|. \qedhere
  \end{equation*}
\end{proof}

\noindent\textbf{Claim 4.} $f(\succsim)=\{x_1\}$.

\begin{proof}
  First, we show that $\sum_{i\in N}s(r(x_1,\succsim_i))>\sum_{i\in N}s(r(x_m,\succsim_i))$.
  By the construction of $\succsim^1$,
  \begin{equation}
    \sum_{k=1}^{m-1}\sum_{i\in N^1}s(r(x_k,\succsim^1_i))=K\sum_{j=1}^{m-1} s(j).
  \end{equation}
  By the condition that $\sum_{i\in N^1}s(r(x_1,\succsim^1_i))>\sum_{i\in N^1}s(r(x_k,\succsim^1_i))$ for all $k\in \{2,3,\dots,m-1\}$,
  \begin{equation}
    \sum_{i\in N^1}s(r(x_1,\succsim^1_i))>\frac{1}{m-1}K\sum_{j=1}^{m-1} s(j).
  \end{equation}
  In addition, for all $\ell\in\{2,\dots,m\}$,
  \begin{equation}
    \sum_{i\in N^\ell}s(r(x_1,\succsim^\ell_i))=\frac{1}{m-1}K \qty[ \sum_{j=1}^m  s(j)-s(m-\ell+1) ].
  \end{equation}
  Therefore,
  \begin{align*}
    \sum_{i\in N\backslash\{i^+\}}s(r(x_1,\succsim_i))
    &>\frac{1}{m-1}K\sum_{j=1}^{m-1} s(j) + \sum_{\ell=2}^m \frac{1}{m-1}K \qty[ \sum_{j=1}^m   s(j)-s(m-\ell+1) ]\\
    &=\sum_{\ell=1}^m \frac{1}{m-1}K \qty[ \sum_{j=1}^m s(j)-s(m-\ell+1) ]\\
    &=\frac{m-1}{m-1}K\sum_{j=1}^m s(j)\\
    &=K\sum_{j=1}^m s(j)\\
    &=\sum_{i\in N\backslash\{i^+\}}s(r(x_m,\succsim_i)).
  \end{align*}
  Since $K$ is sufficiently large, we have $$\sum_{i\in N\backslash\{i^+\}}s(r(x_1,\succsim_i))-\sum_{i\in N\backslash\{i^+\}}s(r(x_m,\succsim_i))>s(1) - s(m).$$
  Thus we have
  \begin{align*}
    \sum_{i\in N}s(r(x_1,\succsim_i))&=\sum_{i\in N\backslash\{i^+\}}s(r(x_1,\succsim_i))+s(m)\\
    &>\sum_{i\in N\backslash\{i^+\}}s(r(x_m,\succsim_i)) +s(1)\\
    &=\sum_{i\in N}s(r(x_m,\succsim_i)).
  \end{align*}

  Next, we show that $\sum_{i\in N}s(r(x_1,\succsim_i))>\sum_{i\in N}s(r(x_k,\succsim_i))$ for all $k\in\{2,3,\dots,m-1\}$. Take any $k\in\{2,3,\dots,m-1\}$. By the definition of $\succsim^1$, we have $\sum_{i\in N^1}s(r(x_1,\succsim^1_i))>\sum_{i\in N^1}s(r(x_k,\succsim^1_i))$.
  Since $K$ is sufficiently large, we have
  \[
    \sum_{i\in N^1}s(r(x_1,\succsim^1_i))-\sum_{i\in N^1}s(r(x_k,\succsim^1_i))>s(m-k+1) - s(m).
  \]
  In addition, $\sum_{i\in N^\ell}s(r(x_1,\succsim^\ell_i))=\sum_{i\in N^\ell}s(r(x_k,\succsim^\ell_i))$ for all $\ell\in\{2,\dots,m\}$.
  Thus, we have
  \begin{align*}
    \sum_{i\in N}s(r(x_1,\succsim_i))
    &=\sum_{i\in N^1}s(r(x_1,\succsim^1_i)) + s(m) +\sum_{\ell=2}^m\sum_{i\in N^\ell}s(r(x_1,\succsim^\ell_i))\\
    &>\sum_{i\in N^1}s(r(x_k,\succsim^1_i))+s(m-k+1)+\sum_{\ell=2}^m\sum_{i\in N^\ell}s(r(x_k,\succsim^\ell_i))\\
    &=\sum_{i\in N}s(r(x_k,\succsim_i)). \qedhere
  \end{align*}
\end{proof}

\noindent\textbf{Claim 5.} $x_1\notin f'(\succsim)$.

\begin{proof}
  It suffices to show that $\sum_{i\in N}s'(r(x_1,\succsim_i))<\sum_{i\in N}s'(r(x_2,\succsim_i))$.
  By the definition of $\succsim^1$, we have $\sum_{i\in N^1}s'(r(x_1,\succsim^1_i))<\sum_{i\in N^1}s'(r(x_2,\succsim^1_i))$.
  In addition, for all $\ell\in\{2,\dots,m\}$, we have $\sum_{i\in N^\ell}s'(r(x_1,\succsim^\ell_i))=\sum_{i\in N^\ell}s'(r(x_2,\succsim^\ell_i))$.
  Thus, we have
  \begin{align*}
    \sum_{i\in N}s'(r(x_2,\succsim_i))
    &=\sum_{i\in N^1}s'(r(x_2,\succsim^1_i)) +\sum_{\ell=2}^m\sum_{i\in N^\ell}s'(r(x_2,\succsim^\ell_i)) +  s'(m-1) \\
    &>\sum_{i\in N^1}s'(r(x_1,\succsim^1_i))+\sum_{\ell=2}^m\sum_{i\in N^\ell}s'(r(x_1,\succsim^\ell_i))+s'(m)\\
    &=\sum_{i\in N}s'(r(x_1,\succsim_i)). \qedhere
  \end{align*}
\end{proof}

\subsubsection{The case where Lemma \ref{lem:convolution}(iii) holds}

Suppose that conditions ave(a)-ave(d) hold.
By the induction hypothesis, there exist a preference profile $\succsim^3 \in \mathcal{D}$ on $\{x, y, z\}$ such that $x$ is a Condorcet loser at $\succsim^3$, $f^{s_{\text{ave}}}(\succsim^3) = \{x\}$, and $x \notin f^{s'_{\text{ave}}}(\succsim^3)$.
Let $n$ be the number of voters in $\succsim^3$.
We partition the voters in $N^3$ into three sets based on the rank of $x$: $n_1$ voters rank $x$ first, $n_2$ voters rank $x$ second, and $n_3$ voters rank $x$ third.
Since $x$ is a Condorcet loser at $\succsim^3$, we have $|\{i \in N^3 : y \succsim_i^3 x\}| > |\{i \in N^3 : x \succsim_i^3 y\}|$ and $|\{i \in N^3 : z \succsim_i^3 x\}| > |\{i \in N^3 : x \succsim_i^3 z\}|$. Summing both sides of these inequalities yields $(2n_3 + n_2) > (2n_1 + n_2)$, which implies $n_3 > n_1$.

Let $X = \{x, y, z, w_1, \dots, w_{m-3}\}$. We construct a profile $\succsim^{\text{ave}}$ on $X$ consisting of two sub-profiles, $\succsim^*$ and $\succsim'$. 

\noindent\textbf{Construction of $\succsim^*$:}
The profile $\succsim^*$ consists of $3n(m-2)(m-2)!$ voters in total.
For each voter $i$ in $N^3$ with preferences $a_i \succ^3_i b_i \succ^3_i c_i$ on $\{x,y,z\}$, we create $3(m-2)(m-2)!$ voters in $\succsim^*$. Specifically, let $\Pi_i$ be the set of all $(m-2)!$ permutations of $\{b_i, w_1, \dots, w_{m-3}\}$. For each permutation $\pi \in \Pi_i$, we introduce $3(m-2)$ voters who rank $a_i$ first, $c_i$ last, and order the remaining alternatives in positions $2$ through $m-1$ according to $\pi$.

\noindent\textbf{Construction of $\succsim'$:}
The profile $\succsim'$ consists of $2n(m-3)(m-1)!$ voters in total.
For each $j \in \{1, \dots, m-3\}$, we create two groups of voters, each of size $n(m-1)!$:
\begin{enumerate}
  \item Let $\hat{\Pi}_j$ be the set of all permutations of $X \setminus \{w_j\}$. For each permutation $\pi \in \hat{\Pi}_j$, we introduce $n$ voters who rank $w_j$ first and rank the remaining alternatives in positions 2 through $m$ according to $\pi$.
  \item For each permutation $\pi \in \hat{\Pi}_j$, we introduce $n$ voters who rank $w_j$ last and rank the remaining alternatives in positions 1 through $m-1$ according to $\pi$.
\end{enumerate}
\noindent\textbf{Claim 6.} $x$ is a Condorcet loser at $\succsim^{\text{ave}} = (\succsim^*, \succsim')$.
\begin{proof}
  By construction, in $\succsim'$, the pairwise majority comparison between $x$ and any other alternative results in a tie. Since $\succsim^*|_{\{x,y,z\}}$ consists of $3(m-2)(m-2)!$ copies of $\succsim^3$ and $x$ is beaten by $y$ and $z$ in $\succsim^3$, $x$ is also beaten by them in $\succsim^*$. For the pairwise comparison between $x$ and $w_j$ in $\succsim^*$, we have $|\{i \in N^* : w_j \succsim^*_i x\}| - |\{i \in N^* : x \succsim^*_i w_j\}| = 3(m-2)(m-2)!(n_3 - n_1)$. Since $n_3 > n_1$, $w_j$ defeats $x$. Thus $x$ is a Condorcet loser at $\succsim^*$, and hence $x$ is a Condorcet loser at $\succsim^{\text{ave}}$.
\end{proof}

\noindent\textbf{Claim 7.} $f(\succsim^{\text{ave}})=\{x\}$.

\begin{proof}
  First, we compare the score of $x$ with those of $y$ and $z$. By construction, we have $\sum_{i \in N'} s(r(x, \succsim'_i))=\sum_{i \in N'} s(r(y, \succsim'_i))=\sum_{i \in N'} s(r(z, \succsim'_i))$.
  In addition,
  \begin{align*}
    \sum_{i \in N^*} s(r(x, \succsim^*_i))&=3(m-2)(m-2)! \qty{n_1s(1)+n_2\frac{1}{m-2}\sum_{k=2}^{m-1}s(k)+n_3s(3)}\\
    &=3(m-2)(m-2)!\sum_{i\in N^3}s(r(x,\succsim^3_i)).
  \end{align*}
  Similarly, we have $\sum_{i \in N^*} s(r(y, \succsim^*_i))=3(m-2)(m-2)!\sum_{i\in N^3}s(r(y,\succsim^3_i))$ and $\sum_{i \in N^*} s(r(z, \succsim^*_i))=3(m-2)(m-2)!\sum_{i\in N^3}s(r(z,\succsim^3_i))$. Since $f^{s_{\text{ave}}}(\succsim^3) = \{x\}$,
  \begin{align*}
    &\sum_{i\in N^3}s(r(x,\succsim^3_i))>\sum_{i\in N^3}s(r(y,\succsim^3_i)) ~~~ \text{and}~~~
    \sum_{i\in N^3}s(r(x,\succsim^3_i))>\sum_{i\in N^3}s(r(z,\succsim^3_i)).
  \end{align*}
  Thus we have
  \begin{align*}
    &\sum_{i \in N^*} s(r(x, \succsim^*_i))>\sum_{i \in N^*} s(r(y, \succsim^*_i)) ~~~\text{and}~~~
    \sum_{i \in N^*} s(r(x, \succsim^*_i))>\sum_{i \in N^*} s(r(z, \succsim^*_i)).
  \end{align*}
  They imply that
  \begin{align*}
    \sum_{i \in N^{\text{ave}}} s(r(x, \succsim^{\text{ave}}_i))&=\sum_{i \in N^*} s(r(x, \succsim^*_i))+\sum_{i \in N'} s(r(x, \succsim'_i))\\
    &>\sum_{i \in N^*} s(r(y, \succsim^*_i))+\sum_{i \in N'} s(r(y, \succsim'_i))\\
    &=\sum_{i \in N^{\text{ave}}} s(r(y, \succsim^{\text{ave}}_i))
  \end{align*}
  and $\sum_{i \in N^{\text{ave}}} s(r(x, \succsim^{\text{ave}}_i)) > \sum_{i \in N^{\text{ave}}} s(r(z, \succsim^{\text{ave}}_i))$.

  Next, take any $j\in\{1,\dots,m-3\}$ and we compare the score of $x$ and $w_j$. By straightforward calculation, we have
  \begin{align}
    &\sum_{i \in N^*} s(r(w_j, \succsim^*_i))=3n(m-2)!\sum_{k=2}^{m-1}s(k) \text{ and}\\
    &\sum_{i \in N'} s(r(w_j, \succsim'_i))=n(m-2)!\qty{(2m-5)s(1)+(2m-8)\sum_{k=2}^{m-1}s(k)+(2m-5)s(m) }.
  \end{align}
  Thus we have
  \begin{align}
    \sum_{i \in N^{\text{ave}}} s(r(w_j, \succsim^{\text{ave}}_i))
    = \sum_{i \in N^*} s(r(w_j, \succsim^*_i))+\sum_{i \in N'} s(r(w_j, \succsim'_i))
    =n(2m-5)(m-2)!\sum_{k=1}^{m}s(k).
  \end{align}
  Since $|N^{\text{ave}}|=|N^*|+|N'|=3n(m-2)(m-2)!+2n(m-3)(m-1)!=mn(2m-5)(m-2)!$, the score of $w_j$ is exactly the average score of all alternatives. Therefore, the sum of the scores of $x$, $y$, and $z$ is exactly three times the average score. Since the score of $x$ is higher than those of $y$ and $z$, it is higher than the average score. Therefore the score of $x$ is higher than that of $w_j$.
\end{proof}

\noindent\textbf{Claim 8.} $x\notin f'(\succsim^{\text{ave}})$.

\begin{proof}
  We show that there exists $v \in \{y, z\}$ whose score is strictly higher than that of $x$. By construction, we have $\sum_{i \in N'} s'(r(x, \succsim'_i))=\sum_{i \in N'} s'(r(y, \succsim'_i))=\sum_{i \in N'} s'(r(z, \succsim'_i))$. In addition,
  \begin{align*}
    &\sum_{i \in N^*} s'(r(x, \succsim^*_i))=3(m-2)(m-2)!\sum_{i\in N^3}s'(r(x,\succsim^3_i)), \\
    &\sum_{i \in N^*} s'(r(y, \succsim^*_i))=3(m-2)(m-2)!\sum_{i\in N^3}s'(r(y,\succsim^3_i)), \text{ and}\\
    &\sum_{i \in N^*} s'(r(z, \succsim^*_i))=3(m-2)(m-2)!\sum_{i\in N^3}s'(r(z,\succsim^3_i)).
  \end{align*}
  Since $x\notin f^{s'_{\text{ave}}}(\succsim^3)$, we have $\sum_{i\in N^3}s'(r(x,\succsim^3_i))<\sum_{i\in N^3}s'(r(y,\succsim^3_i))$ or $\sum_{i\in N^3}s'(r(x,\succsim^3_i))<\sum_{i\in N^3}s'(r(z,\succsim^3_i))$. Without loss of generality, we assume $\sum_{i\in N^3}s'(r(x,\succsim^3_i))<\sum_{i\in N^3}s'(r(y,\succsim^3_i))$.
  Then, we have $\sum_{i\in N^*}s'(r(x,\succsim^*_i))<\sum_{i\in N^*}s'(r(y,\succsim^*_i))$.
  Therefore, we have
  \begin{align*}
    \sum_{i \in N^{\text{ave}}} s'(r(x, \succsim^{\text{ave}}_i))&=\sum_{i \in N^*} s'(r(x, \succsim^*_i))+\sum_{i \in N'} s'(r(x, \succsim'_i))\\
    &<\sum_{i \in N^*} s'(r(y, \succsim^*_i))+\sum_{i \in N'} s'(r(y, \succsim'_i))\\
    &=\sum_{i \in N^{\text{ave}}} s'(r(y, \succsim^{\text{ave}}_i)). \qedhere
  \end{align*}
\end{proof}

\subsubsection{The case with $s(2)=s(m-1)=\frac{1}{2}$ and $s'(2)=s'(m-1)=\alpha<\frac{1}{2}$}
For readability, let $X=\{x,y_1,y_2,\ldots,y_{m-1}\}$. We construct a profile that satisfies three conditions of Proposition \ref{prop:existence}. Let $b\in\mathbb{N}$ be such that
\[
  \frac{1-\alpha}{1-2\alpha}<b.
\]
Let $\succsim\in\mathscr{P}^{(m-1)b+1}$ be such that
\begin{align*}
  &|\{i\in N:y_1\succsim_iy_2\succsim_i \cdots\succsim_i y_{m-2} \succsim_i x\succsim_i y_{m-1}\}|=b,\\
  &|\{i\in N:y_2\succsim_iy_3\succsim_i \cdots\succsim_i y_{m-1} \succsim_i x\succsim_i y_{1}\}|=b,\\
  &|\{i\in N:y_3\succsim_iy_4\succsim_i \cdots\succsim_i y_{1} \succsim_i x\succsim_i y_{2}\}|=b,\\
  &\hspace{10pt}\vdots \\
  &|\{i\in N:y_{m-1}\succsim_iy_1\succsim_i \cdots\succsim_i y_{m-3} \succsim_i x\succsim_i y_{m-2}\}|=b, \text{ and}\\
  &|\{i\in N:x\succsim_iy_1\succsim_i y_2 \succsim_i \cdots \succsim_i y_{m-1}\}|=1.
\end{align*}
Since $|\{i\in N: x\succsim_i y_k\}|=b+1<(m-2)b=|\{i\in N: y_k\succsim_i x\}|$ holds for all $k\in \{1,\dots,m-1\}$, $x$ is a Condorcet loser at $\succsim$.

By
\begin{align*}
  &\sum_{i\in N}s(r(x,\succsim_i))=1\cdot 1+\frac{1}{2}\cdot (m-1)b = \frac{m-1}{2}b+1,\\
  &\sum_{i\in N}s(r(y_k,\succsim_i))=1\cdot b+\frac{1}{2}\cdot ((m-3)b+1) = \frac{m-1}{2}b+\frac{1}{2} \qquad(k=1,2,\dots,m-2), ~\text{ and}\\
  &\sum_{i\in N}s(r(y_{m-1},\succsim_i))=1\cdot b+\frac{1}{2}\cdot (m-3)b = \frac{m-1}{2}b,
\end{align*}
we have $f(\succsim)=\{x\}$.

By
\begin{align*}
  &\sum_{i\in N}s'(r(x,\succsim_i))=1\cdot 1+\alpha\cdot (m-1)b = (m-1)\alpha b +1~\text{ and}\\
  &\sum_{i\in N}s'(r(y_1,\succsim_i))=1\cdot b+\alpha\cdot ((m-3)b+1) = ((m-3)\alpha+1) b+\alpha,
\end{align*}
we have $\sum_{i\in N}s'(r(y_1,\succsim_i))-\sum_{i\in N}s'(r(x,\succsim_i))=b(1-2\alpha)-(1-\alpha) > 0$.
Thus, $x\notin f'(\succsim)$.

\subsubsection{The case with $s(2)=s(m-1)=\frac{1}{2}$ and $s'(2)=s'(m-1)=\alpha>\frac{1}{2}$}
For readability, let $X=\{x,y_1,y_2,\ldots,y_{m-1}\}$. We construct a profile that satisfies three conditions of Proposition \ref{prop:existence}. Let $b,c\in\mathbb{N}$ be such that
\begin{align*}
  &b<c<\frac{2m-4}{m-1}b,\text{ and}\\
  &c<\frac{2+(m-4)\alpha}{m-1-(m-2)\alpha}b.
\end{align*}
Note that since $m \ge 4$ and $\alpha > \frac{1}{2}$, we have $\frac{2m-4}{m-1} > 1$ and $\frac{2+(m-4)\alpha}{m-1-(m-2)\alpha} > 1$.
Let $\succsim\in\mathscr{P}^{2(m-1)b+(m-1)c}$ be such that
\begin{align*}
  &|\{i\in N:y_1\succsim_iy_2\succsim_i \cdots\succsim_i y_{m-2} \succsim_i x\succsim_i y_{m-1}\}|=b,\\
  &|\{i\in N:y_2\succsim_iy_3\succsim_i \cdots\succsim_i y_{m-1} \succsim_i x\succsim_i y_{1}\}|=b,\\
  &\hspace{10pt}\vdots \\
  &|\{i\in N:y_{m-1}\succsim_iy_1\succsim_i \cdots\succsim_i y_{m-3} \succsim_i x\succsim_i y_{m-2}\}|=b,\\[10pt]
  &|\{i\in N:y_1\succsim_iy_2\succsim_i \cdots\succsim_i y_{m-2} \succsim_i y_{m-1}\succsim_i x\}|=b,\\
  &|\{i\in N:y_2\succsim_iy_3\succsim_i \cdots\succsim_i y_{m-1} \succsim_i y_1\succsim_i x\}|=b,\\
  &\hspace{10pt}\vdots \\
  &|\{i\in N:y_{m-1}\succsim_iy_1\succsim_i \cdots\succsim_i y_{m-3} \succsim_i y_{m-2}\succsim_i x\}|=b,\\[10pt]
  &|\{i\in N:x \succsim_i y_1\succsim_iy_2\succsim_i \cdots\succsim_i y_{m-1}\}|=c,\\
  &|\{i\in N:x \succsim_i y_2\succsim_iy_3\succsim_i \cdots\succsim_i y_{1}\}|=c,\\
  &\hspace{10pt}\vdots \\
  &|\{i\in N:x \succsim_i y_{m-1}\succsim_iy_1\succsim_i \cdots\succsim_i y_{m-2}\}|=c.
\end{align*}
Since $|\{i\in N: x\succsim_i y_k\}|=b+(m-1)c<b+(2m-4)b=(2m-3)b=|\{i\in N: y_k\succsim_i x\}|$ holds for all $k\in\{1,\dots,m-1\}$, $x$ is a Condorcet loser at $\succsim$.

By
\begin{align*}
  &\sum_{i\in N}s(r(x,\succsim_i))=1\cdot (m-1)c+\frac{1}{2}\cdot (m-1)b = \frac{m-1}{2}b+(m-1)c ~\text{ and}\\
  &\sum_{i\in N}s(r(y_k,\succsim_i))=1\cdot 2b+\frac{1}{2}\cdot ((2m-5)b+(m-2)c) = \frac{2m-1}{2}b+\frac{m-2}{2}c \quad (k=1,2,\dots,m-1),
\end{align*}
we have $\sum_{i\in N}s(r(x,\succsim_i))-\sum_{i\in N}s(r(y_k,\succsim_i))=\frac{m}{2}(c-b)>0$, and hence $f(\succsim)=\{x\}$.

By
\begin{align*}
  &\sum_{i\in N}s'(r(x,\succsim_i))=1\cdot (m-1)c+\alpha\cdot (m-1)b = (m-1)\alpha b+(m-1)c ~\text{ and}\\
  &\sum_{i\in N}s'(r(y_1,\succsim_i))=1\cdot 2b+\alpha\cdot ((2m-5)b+(m-2)c) = 2b+(2m-5)\alpha b +(m-2)\alpha c ,
\end{align*}
we have $\sum_{i\in N}s'(r(y_1,\succsim_i))-\sum_{i\in N}s'(r(x,\succsim_i))=(2+(m-4)\alpha)b-(m-1-(m-2)\alpha)c>0$, and hence $x\notin f'(\succsim)$.

\bibliographystyle{econ}
\bibliography{reference}

\end{document}